\documentclass[review]{elsarticle}
\usepackage{amsthm}
\usepackage{amsmath}

\usepackage{lipsum}
\makeatletter
\def\ps@pprintTitle{%
	\let\@oddhead\@empty
	\let\@evenhead\@empty
	\def\@oddfoot{}%
	\let\@evenfoot\@oddfoot}
\makeatother

\newtheorem{df}{Definition}[section]
\newtheorem{tm}{Theorem}[section]
\newtheorem{obs}{Observation}[section]
\newtheorem{lm}{Lemma}[section]

\newtheorem{cor}{Corollary}[section]
\usepackage{amsmath,amsfonts,amssymb,mathtools,mathrsfs}
\usepackage{lineno,hyperref}

\usepackage{subcaption}
\modulolinenumbers[5]

\journal{SODA22}

\usepackage{xcolor}
\usepackage{textcomp}

\newcommand{\ignore}[1]{}

\begin{document}

\begin{frontmatter}

\title{Anti Tai Mapping for Unordered Labeled Trees}

\author[osijek]{Mislav Bla\v{z}evi\'c}
\author[lmu]{Stefan Canzar}
\author[uae]{Khaled Elbassioni}
\author[osijek]{Domagoj Matijevi\'c}


\address[osijek]{Department of Mathematics, University of Osijek, Croatia}
\address[lmu]{Gene Center, Ludwig-Maximilians-Universit\"a{}t M\"u{}nchen, 81377 Munich, Germany}
\address[uae]{Khalifa University of Science and Technology, Abu Dhabi, UAE}

\begin{abstract}
The well-studied Tai mapping between two rooted labeled trees $T_1(V_1, E_1)$ and $T_2(V_2, E_2)$ defines a one-to-one mapping between nodes in $T_1$ and $T_2$ that preserves ancestor relationship \cite{tai79}. 
For unordered trees the problem of finding a maximum-weight Tai mapping is known to be \mbox{NP-complete}~\cite{ZHANG1992133}. 
 In this work, we define an anti Tai mapping $M\subseteq V_1\times V_2$ as a binary relation between two unordered labeled trees such that 
any two $(x,y), (x', y')\in M$ violate ancestor relationship and thus cannot be part of the same Tai mapping,  
i.e. $(x\le x' \iff y\not \le y') \vee (x'\le x \iff y'\not \le y)$, given an ancestor order $x<x'$ meaning that $x$ is an ancestor of $x'$. Finding a maximum-weight anti Tai mapping arises in the cutting plane method for solving the maximum-weight Tai mapping problem via integer programming.
We give an efficient polynomial-time 
algorithm for finding a maximum-weight anti Tai mapping for the case when one  of the two trees is a path and 
further show how to 
extend this result in order to provide a polynomially computable  lower bound on the optimal anti Tai mapping for two unordered labeled trees. 
The latter result stems from the special class of anti Tai mapping defined by the more restricted condition $x\sim x' \iff y\not\sim y'$, where $\sim$ denotes 
 that two nodes belong to the same root-to-leaf path. For this class, we give an efficient algorithm that solves the problem directly on two unordered trees in $O(|V_1|^2|V_2|^2)$. 
\end{abstract}


\end{frontmatter}


\section{Introduction}
The systematic comparison of hierarchies is a fundamental task in a wide range of domains. Entities in biology often stand in a hierarchical relationship to one another (\cite{https://doi.org/10.1002/bimj.201100186}, \cite{Dutkowski2013}). A phylogenetic tree, for example, groups all descendants of a common ancestor and matching metrics are used in order to evaluate 
the dissimilarity between such trees  (\cite{bogdanovicz}). The most common way to measure the distance between trees is the tree edit distance (\cite{tai79}, \cite{YoshinoH17}). In this work we consider unordered labeled trees in which the order among siblings is irrelevant. 
The tree edit distance between trees is defined by the minimal-cost sequence of node edit operations (insert, delete, relabel) that transforms one tree into another. It is  shown in \cite{tai79} that there is a strong relationship between a Tai mapping and a sequence of editing operations, i.e. given a Tai mapping the tree edit distance can easily be computed.

\begin{figure}[t]
\centering
\begin{subfigure}{.45\textwidth}
\includegraphics[width=.93\linewidth]{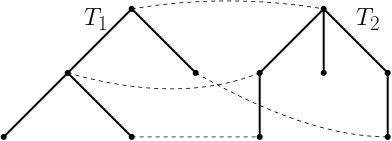}
\caption{Tai mapping.}
\label{fig:tai}
\end{subfigure}
\begin{subfigure}{.45\textwidth}
\includegraphics[width=.93\linewidth]{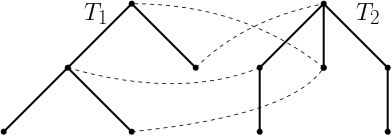}
\caption{Anti Tai mapping.}
\label{fig:antitai}
\end{subfigure}
\caption{Examples for Tai and anti Tai mapping. Dashed lines denote Tai and anti Tai edges. }
\label{fig:taiVSantitai}
\end{figure}

A Tai mapping defines a one-to-one node mapping between trees that preserves ancestor relationships (see Figure~\ref{fig:tai}). 
We use $<$ and $\le$ to denote ancestor orders, i.e. $x<x'$ if $x$ is an ancestor of $x'$, and $x\le x'$ if $x<x'$ or $x=x'$. 
Formally, for any two trees $T_1=(V_1, E_1), T_2(V_2, E_2)$, a Tai mapping is defined as a mapping 
$M\subseteq V_1\times V_2$ such that for any distinct $(x, y), (x', y')\in M$,  we have that 
\begin{equation}\label{tai}
( (x = x') \iff (y=y') ) \land ( (x < x') \iff (y < y') ).
\end{equation}
Equation~\ref{tai} combines  one-to-one and ancestor order property with the conjunction. 
Note that a Tai mapping
can equivalently be defined combining both properties, i.e. for any distinct $(x, y), (x', y')\in M$ we have that
\begin{equation}\label{tai2}
 (x \le x') \iff (y \le y').
\end{equation}
For a weight function $w:V_1\times V_2\to Z^+$ the weight of $M$ is defined as $w(M) := \sum_{(x, y)\in M} w(u, v)$. 
The problem of finding a maximum weight Tai mapping  
is known to be  NP-complete (\cite{ZHANG1992133}) and MAX SNP-hard (\cite{10.1016/0020-0190(94)90062-0}). 
Hence, different integer linear programming (ILP) based algorithms have been proposed to compute such a mapping (\cite{kondo}, \cite{trajan}, \cite{generalizedRF}, \cite{hong2017improved}).         
A naive ILP formulation~\cite{generalizedRF} contains a constraint for each pair of edges $(x,x'), (y, y')\in V_1\times V_2$ such (\ref{tai2}) does not hold, giving rise to $O(nm)$ variables and $O(n^2m^2)$ constraints, where $n$ and $m$ are the number of nodes of the two input trees, and thus does not allow to practically solve even moderate-sized instances (\cite{kondo}, \cite{generalizedRF}).
In \cite{hong2017improved} the authors therefore divide the problem 
into $O(nm)$ subproblems with $O(n+m)$ constraints each, by utilizing a dynamic programming approach from~\cite{Fukagawa2011}. 
In \cite{trajan}, on the other hand, the authors introduced two classes of valid inequalities of the Tai mapping polytope, crossing edges clique constraints and semi-independent clique constraints, and use them as cutting planes in a branch-and-bound scheme. Their implementation achieved a 13-fold speed-up compared to the naive ILP in experiments on perturbed human single-cell data. 
Here, we derive a class of valid inequalities that generalize both types of previously introduced clique constraints.  We formalize them as
special classes of \emph{anti Tai mappings} defined as follows. 

\paragraph{Anti Tai mapping} A mapping $M\subseteq V_1\times V_2$ that consists of edges for which (\ref{tai}) does not hold is called an anti Tai mapping (see Figure~\ref{fig:antitai}). More formally, an anti Tai mapping is defined as a binary relation\footnote{Note that anti Tai mapping is not really a mapping any more.
Hence, it should be understood as anti "Tai mapping".} 
$M\subseteq V_1\times V_2$ if for any distinct $(x, y), (x', y')\in M$ we have that
	\begin{equation}\label{antiTai}
	((x \le x') \iff (y \not\le y') ) \vee ((x' \le x) \iff (y' \not\le y) ).
	\end{equation}
Let $G=(V_1\times V_2,E)$ be a graph on vertex set $V_1\times V_2$ such that there is an edge $\{(x, y), (x', y')\}\in E$ if and only if (\ref{antiTai}) holds. Then a maximum-weight independent set in $G$ corresponds to a maximum-weight Tai mapping, while a maximum clique corresponds to a maximum-weight anti Tai mapping. 
Recall that the stable set polytope, denoted by STAB$(G)$, is the convex hull of stable (independent) sets of $G$ and that it is contained in the fractional stable set polytope
\begin{equation}\label{qstab}
\text{QSTAB}(G)=\{x\in\mathbb{R}^{V_1\times V_2}_+:~\sum_{v\in Q}x_v\le 1\text{ for all cliques }Q\text{ in }G\},
\end{equation}
see \cite{GLS93}. It is known that STAB$(G)$=QSTAB$(G)$ if and only if the graph $G$ is perfect, in which case linear optimization problems over STAB$(G)$ can be solved in polynomial time \cite{GLS93}. As the problem of finding a maximum-weight Tai mapping is NP-hard \cite{ZHANG1992133}, it follows that, unless P=NP, the graph $G$ defined above is not perfect and furthermore that STAB$(G)\subset$QSTAB$(G)$. Heuristically, one can still use the clique constraints in (\ref{qstab}) in a cutting plane algorithm for solving an integer program for the maximum-weight Tai mapping problem, if it is possible to efficiently solve the corresponding separation problem, which amounts to finding a maximum-weight anti Tai mapping. While it is not known whether such a separation problem is polynomially solvable in general, we manage to answer this question in the affirmative when one of the two trees  is a path.
Specifically, we construct a dynamic program to compute the maximum-weight anti Tai mapping for the case where one tree is a path in time $O(|V_1||V_2|)$. Furthermore, we define
a semi-independent antimatching (si-antimatching) as a restricted class of an anti Tai mapping and use this to provide a polynomially computable lower bound on the maximum-weight anti Tai mapping in the case of two trees. More precisely, 
let $x\sim x'$ denote that $x$ and $x'$ are on the same root-to-leaf path (``comparable''), i.e. $x\sim x'$ if $x\leq x'$ or $x'\leq x$.
Otherwise, we say that $x$ and $x'$ are incomparable. A mapping $M\subseteq V_1\times V_2$ is a si-antimatching if for any distinct
$(x,y), (x', y')\in M$:
\begin{equation}\label{si-antimatching}
(x \sim x') \iff (y \not\sim  y')
\end{equation}
Note that any pair $(x,y), (x', y')$ satisfying (\ref{si-antimatching}) also satisfies (\ref{antiTai}), and hence the set of all such pairs form a subgraph of the graph $G$ defined above. It follows that a maximum-weight si-antimatching can be used to provide a valid clique constraint for the Tai mapping problem. Motivated by this, we introduce a dynamic program that optimally solves the maximum-weight si-antimatching problem in time $O(|V_1|^2|V_2|^2)$.  We observe further that the same dynamic program can be combined with our optimal path-tree anti Tai mapping algorithm mentioned above (that is, when of the two trees is a path) to obtain a polynomially computable lower bound on the maximum-weight anti Tai mapping in the general  case.

The rest of the paper is structured as follows. In Section \ref{s2}, we study the properties of an si-antimatching and show that it induces a partition of the matched vertices in each tree into a path and an independent set such that the path in one tree is mapped (by the antimatching) to the independent set in the other and vice versa. We use this structural result to derive a dynamic program for computing a maximum-weight si-antimatching. In Section~\ref{s3}, we 
 give a dynamic programming formulation for the anti Tai mapping problem between a tree and a path, and combine this result with our dynamic programming idea  from Section~\ref{s2} to give a dynamic program that provides a lower bound on the maximum-weight anti Tai mapping between two trees. We conclude in  Section \ref{s5}.

\section{Si-antimatching problem}\label{s2}

In this section, we present an efficient algorithm for this special case of anti Tai mapping. Surprisingly, si-antimatching turns out to be solvable in polynomial time and we provide a dynamic program that solves the problem directly on two unordered labeled trees $T_1$ and $T_2$ in time $O(|V_1|^2|V_2|^2)$. Our dynamic program strongly relies on the decomposition theorem (Theorem~\ref{decomposition_theorem}) that states that every si-antimatching can be decomposed into an antichain and a path in $T_1$, and a path and an antichain in $T_2$ such that an antichain in $T_1$ maps into a path in $T_2$, and a path in $T_1$ maps into an antichain in $T_2$ (see Figure~\ref{fig:decomposition}). We further argue that there exists an order in that decomposition which in turn provides the way of computing an entry in a dynamic table. 
\subsection{Decomposition theorem}
We will first show that si-antimatching can get nicely decomposed, which will form  
the basis for our algorithm.

\begin{df}
Let $T_1(V_1, E_1)$ and $T_2(V_2, E_2)$ denote two trees. We say that  any two
$(x,x'), (y, y') \in V_1\times V_2$ are semi-independent pair of edges (si-edges) if $x\sim x'\iff y\not\sim y'$ holds true. 
\end{df}
Let 
$M\subseteq V_1\times V_2$
be a set of semi-independent edges between trees $T_1$ and $T_2$. Since here we allow
edges with common vertices (i.e. it is not a matching set), we will refer to 
any such set $M$ as semi-independent antimatching or, in short, {\it si-antimatching}\footnote{
Typically in an antimatching, any pair of distinct edges have a common endpoint. 
Note that we use the notion of antimatching in a slightly different way, i.e. 
any pair of distinct edges could have a common endpoint.}.

For any $A\subseteq V_1 $ and {\it si-antimatching} $M$ let 
$M(A):=\{y\in V_2: x\in A \textrm{ and } (x, y)\in M\}$. If $A$ is a singleton, e.g. $A=\{x\}$,  we will often 
omit the set notation and write $M(\{x\}) = M(x)$. Analogously, for any $B\subseteq V_2$ 
let $M^{-1}(B):=\{ x\in V_1: y\in B \textrm{ and } (x, y)\in M \}$. 
Furthermore, let 
$V_1^M\subseteq V_1$ and $V_2^M\subseteq V_2$ denote all nodes in $V_1$ and $V_2$  that are incident to si-edges in $M$. 
We  say that root-to-leaf path $P_1$ in tree $T_1$ is {\it maximal} if no other root-to-leaf path contains more nodes of $V_1^M$ than $P_1$. Let $P_1^M=P_1\cap V_1^M$. 

\begin{figure}[t]
\centering
\includegraphics[width=0.7\textwidth]{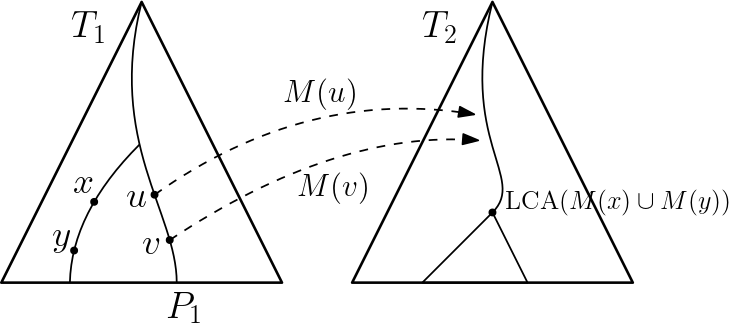}
\caption{Proof of Lemma~\ref{lm:antichain}.}
\label{fig:antichain}
\end{figure}

\begin{lm}\label{lm:antichain} Let $M$ denote an si-antimatching and let $P_1$ denote a maximal path in $T_1$. Then $V_1^M\setminus P_1^M$ is an antichain, i.e. 
any two distinct nodes in $V_1^M\setminus P_1^M$ are incomparable. 
\end{lm}
\begin{proof}
Suppose the opposite, i.e. there exist $x, y\in V_1^M\setminus P_1^M$, $x\not = y$, 
such that $x\sim y$ (see Figure~\ref{fig:antichain}). Since $P_1$ is maximal, there must exist  
$u, v\in P_1^M $ such that $u\not\sim x$, $u\not\sim y$ and $v\not\sim x$, $v\not\sim y$. 
Thus, nodes $M(u)$ and $M(v)$ must be comparable with nodes in $M(x)$ and $M(y)$. Hence, 
it must be that $M(u), M(v) \subseteq [\textrm{root}(T_2), \textrm{LCA}(M(x)\cup M(y))] $, where the 
notation $[\cdot, \cdot]$ denotes a path in a tree with a given start and end node, while LCA stands for a lowest common ancestor of a given set of nodes. Which is a contradiction with the fact that $M(u)$ and $M(v)$ 
cannot lie on a common path in $T_2$. 
\end{proof}
\begin{lm}\label{lm:antichain2path}
Let $M$ be si-antimatching and $A\subseteq V_1^M$ be an antichain. Then either $M(A)$ or $M(A\setminus \{x\})$, for some $x\in A$, lie on a single root-to-leaf path in $T_2$. 
\end{lm}
\begin{proof}
Suppose nodes in $M(A)$ do not lie on a path, i.e. there exist $u, v \in M(A)$ such that $u\not\sim v$. 
Note that $M^{-1}(\{u, v\})\cap A$ contains only a single node. Let $x$ denote that unique node. Since 
$x\not\sim y$ for all $y\in A\setminus \{x\}$ it follows that $u\sim z$ and $v\sim z$ for all 
$z\in M(A\setminus \{x\})$. Thus, all nodes in $M(A\setminus \{x\})$ must lie on the path
$[root(T_2), LCA({u, v})]$. 
\end{proof}

\ignore{
\begin{figure}[t]
\centering
\includegraphics[width=0.7\textwidth]{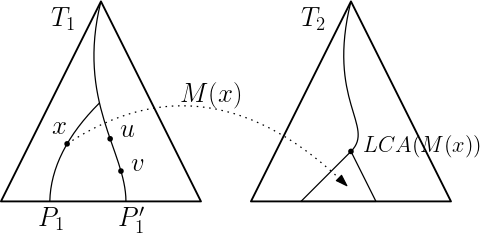}
\caption{Here comes the description of the image}
\label{fig:antichain2path}
\end{figure}
}
\begin{lm}\label{lm2:antichain2path}
Let $M$ denote si-antimatching. Then there exist a maximal path $P_1$ such that all nodes in 
$M(V_1^M\setminus P_1^M)$ lie on a root-to-leaf path in $T_2$. 
\end{lm}
\begin{proof}
Let $P_1'$ be an arbitrary maximal path and $A_1 = V_1^M\setminus P_1'$ an antichain 
(by Lemma~\ref{lm:antichain}). Suppose not all nodes in $M(A_1)$ lie on a path, i.e. there exist $x\in A_1$
such that nodes in $M(x)$ do not all lie on a single path in $T_2$.
Then the path $P_1$ containing $x$ is also maximal since otherwise there would exist 
$u,v\in P_1'\cap V_1^M$, $u\neq v$, $u \not\sim x$ and $v\not\sim x$, such that 
 $M(u), M(v)\subseteq [root(T_2), LCA(M(x))]$, which is a contradiction to the fact that $u\sim v$. 
Hence, there exists a single node $u\in P_1'\cap V_1^M$ such that $u\not\sim z$ for all $z\in A_1$. 
Therefore, $A_1\cup \{u\}$ is also an antichain and by Lemma~\ref{lm:antichain2path} it follows that
$M(A_1\cup\{u\}\setminus \{x\})$ is a path and the claim of the lemma follows. 
\end{proof}

\begin{figure}[t]
\centering
\includegraphics[width=0.7\textwidth]{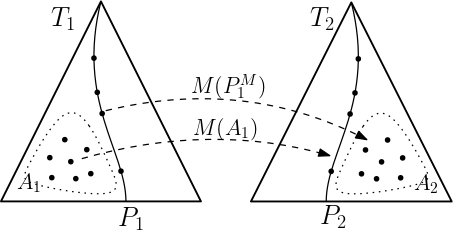}
\caption{Illustration of the decomposition theorem. Nodes on the path $P_1$ will map to the antichain $A_2$, while the antichain $A_1$ will map to nodes on the path $P_2$.}
\label{fig:decomposition}
\end{figure}

Now we are ready to state the main result of this section that immediately follows from previous lemmas: 
\begin{tm}[Decomposition theorem]\label{decomposition_theorem}
For any si-antimatching $M$ there exists a partition of $V_1^M$ into sets $P_1^M = P_1\cap V_1^M$, 
for some root-to-leaf path $P_1$, and $A_1 = V_1^M\setminus P_1$ an antichain, such that $A_2 = M(P_1^M)$ 
is an antichain in $T_2$ and $P_2^M = M(A_1)$ all lie on some root-to-leaf path $P_2$ in $T_2$. 
\end{tm}

\subsection{Dynamic-Programming}
In the following, we will explain our algorithm for the si-antimatching problem based on the decomposition theorem~\ref{decomposition_theorem}. 
Given si-antimatching $M$ and for $x\in A_1$ let $l(x)\in P_1$ denote the lowest ancestor of $x$ in $P_1$. Analogously denote for $x\in A_2$. We start by partitioning antichains into equivalence classes. 
\begin{df}
Partition $A_1$ into equivalence classes $\mathcal{B} = \{B_1, B_2, \ldots, B_n\}$ such that $x, y\in A_1$ belong 
to the same equivalence class if $l(x) = l(y)$. Analogously partition $A_2$ into equivalence classes 
$\mathcal{C} = \{C_1, C_2, \ldots, C_m\}$. 
\end{df}
Without loss of generality we will assume that $l(b_1)<l(b_2) < \ldots < l(b_n)$, for $b_1\in B_1, \ldots, b_n\in B_n$, and similarly $l(c_1)<l(c_2)<\ldots < l(c_m)$, for $c_1\in C_1, \ldots, c_m\in C_m$ (see Figure~\ref{fig:partition}).
Furthermore, for $x\in P_1^M$ let $L(x) = \{l(c_j): C_j\cap M(x) \neq \emptyset \}$.
Given the above decomposition of our problem, with the following theorem we show that we have an order that would allow for dynamic programming to be used. 
\begin{tm}\label{thm2} Let $M$ denote si-antimatching, $P_1^M, A_1$ and $P_2^M, A_2$ partitioning of $V_1^M$ and 
$V_2^M$, respectively, and $\mathcal{B}$ and $\mathcal{C}$ equivalence classes of $A_1$
and $A_2$, respectively. Then for  $x\in P_1^M$
$$
l(b_i) < x  \iff u\leq v,
$$
for all $u\in M(B_i)$  and $v\in L(x)$. 
\end{tm}
\begin{proof}
$[if]$ Let's  assume that for some $x\in P_1^M$, there exists $B_i$ such that 
$u\le v$ for all $u\in M(B_i), v\in L(x)$, and  $l(b_i)\ge x$. That implies that there must exist some $y\in B_i$ such that 
$z\in M(y), w\in M(x)$ and $z\sim w$. Note that it holds that $x\sim y$, since we assumed that $l(b_1)\ge x$, which implies a contradiction with the fact that $M$ is si-antimatching. \\
$[only\mbox{ } if]$ Let $l(b_i)<x$ for some $B_i$ and $u>v$ for all $u\in M(B_i), v\in L(x)$. But then there 
exists $y\in B_i$ such that for $z\in M(y)$, $w\in M(x)$ we have that $z\not\sim w$. Note
that it holds that  $x\not\sim y$, since we assumed that $l(b_i) < x$, which again implies a contradiction with 
the fact that $M$ is si-antimatching. 
\end{proof}


\begin{figure}[t]
\centering
\includegraphics[width=0.6\textwidth]{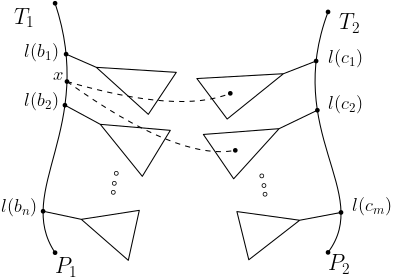}
\caption{Illustration of a path and associated antichains (equivalence classes) attached to it in both trees. }
\label{fig:partition}
\end{figure}

We now see that any si-antimatching has the structure illustrated in Figure~\ref{fig:dp}. Particularly, if we separate elements of $P_1^M$ into equivalence classes of relation $\{(p_i,p_j) : p_i \le l(b_k) \iff p_j \le l(b_k), \forall k=1,...,n\}$ then for $x,y$ from distinct equivalence classes we have $
x < y  \iff u< v,
$
for all $u\in L(x)$  and $v\in L(y)$. Analogous claim holds for $P_2^M$. In other words, images of distinct equivalence classes under $M$ are not interwoven, allowing for a dynamic programming approach to be used.

\begin{figure}[t]
\centering
\includegraphics[width=0.5\textwidth]{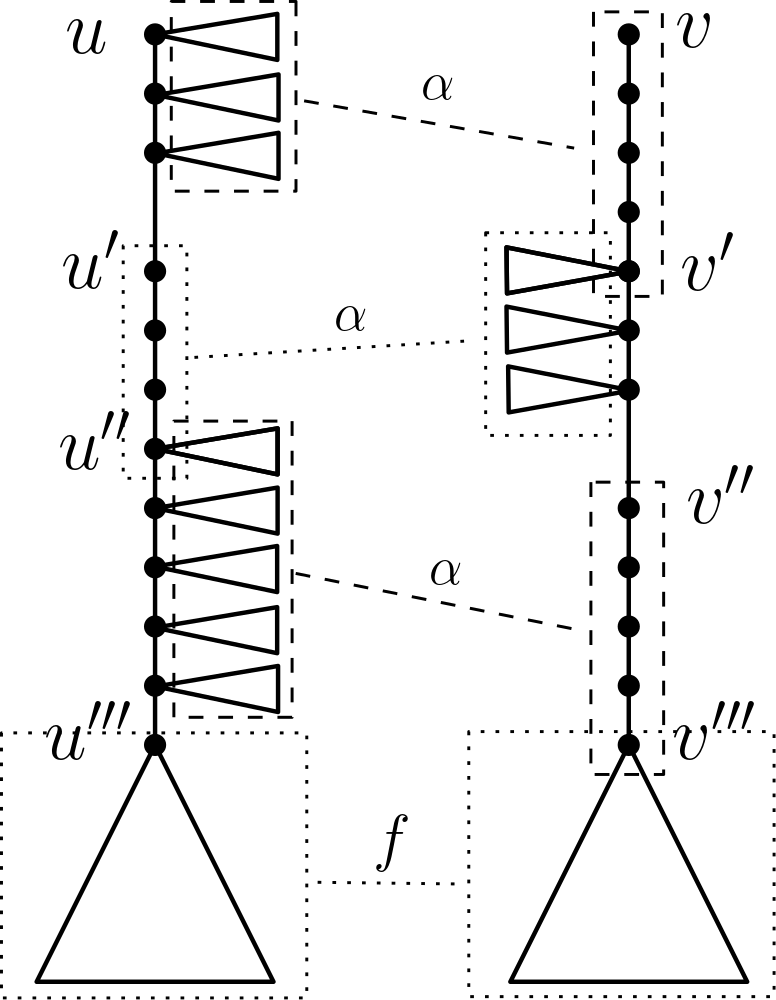}
\caption{An si-antimatching is obtained by alternating between selecting a subpath and a family of subtrees along the paths from the decomposition theorem.}
\label{fig:dp}
\end{figure}

\begin{df}
  We will denote the set of vertices of the subtree rooted in the vertex $x$ with $\tau(x)$, the set of children of the vertex $x$ with $c(x)$, the root of a tree $T$ with $r(T)$ and the unique parent of vertex $x\ne r(T)$ with $p(x)$.
\end{df}

Before stating the main result, note the following.

\begin{lm}\label{wlogl}
Let $P_1^M=\{p_1,...,p_k\}$, $P_2^M=\{q_1,...,q_r\}$, w.l.o.g. $p_1<p_2<...<p_k$, $q_1<q_2<...<q_r$.
Then either $l(b_1)<p_1$ or $l(c_1)<q_1$.
\proof
Assume both $l(b_1)\ge p_1$ and $l(c_1)\ge q_1$. Then $M(p_1)\subseteq \tau(q_1)$ and $M^{-1}(q_1)\subseteq \tau(p_1)$ which is a contradiction with the assumption that $M$ is an si-antimatching.
\qed
\end{lm}

\begin{tm}\label{thmdp}
    Let $\alpha(x, v, v')$ be the weight of maximum weight antichain in the subtree rooted at vertex $x$ with weight $\sum_{z\in [v,v']}w(y, z)$ on vertex $y\in \tau (x)$. Then,

  $\displaystyle\max_{M \text{si-antimatching}}w(M)=\max \begin{cases}f(r(T_1), r(T_2)) \\ f(r(T_2), r(T_1))\end{cases}$, where
  $$f(u,v)=\max\begin{cases}\displaystyle\max_{\substack{u'\in \tau (u)\setminus \{u\} \\ v'\in \tau (v)\setminus \{v\}}}\left[ f(v', u') + \sum_{x\in [u,p(u')]} \sum_{y\in c(x) \setminus [u,u']} \alpha(y, v, v')\right] \\ \displaystyle\max_{v'\in \tau(v)} \alpha (u, v, v')\end{cases}$$
  \proof
  Let $u'\in \tau(u)\setminus \{u\}$ and $v'\in \tau(v)\setminus \{v\}$. The union $\mathcal{A}$ of antichains of any family of subtrees rooted in $c([u,p(u')])\setminus \{u'\}$, is an antichain and so any subset of $\mathcal{A}\times [v,v']$ is an si-antimatching. Since we have $x\not\sim y$ for all $x\in\tau(u')$, $y\in \mathcal{A}$ and $x\sim y$ for all $x\in [v,v']$, $y\in \tau (v')$, inductively any set found by the dynamic program is an si-antimatching. Conversely, let $M$ be a non-empty si-antimatching.
If $P_1^M=\emptyset$, we have $w(M)\le \alpha (r(T_1), q_1, q_r)$, where $q_1$ and $q_r$ are respectively the least and the greatest element of $P_2^M$. If $P_1^M = \{p_1,...,p_k\}$, w.l.o.g. $p_1<p_2<...<p_k$. By lemma \ref{wlogl} it is not a loss of generality to assume $l(b_1)<p_1$. Partition $M=G\cup F$ where $G=M\cap((V_1 \setminus \tau(p_1))\times M(V_1 \setminus \tau(p_1)))$ and $F=M\setminus G$. We have $w(G) \le \sum_{x\in [r(T_1), p(p_1)]} \sum_{y\in c(x)\setminus [r(T_1),p_1]} \alpha (y, q_1, q_r)$ where $q_1$ and $q_r$ are respectively the least and the greatest element of $M(G)$. By proceeding analogously with the subtrees rooted in $q_r$ and $p_1$ and the si-antimatching $F$, we obtain $w(F)\le f(q_r, p_1)$.

  \qed
\end{tm}

It should be noted that it is possible to implement the algorithm from theorem \ref{thmdp} with running time
$O(|V_1|^2|V_2|^2)$. Specifically, for a fixed $v'\in \tau(v)\setminus\{v\}$ the maximum weight antichains rooted in $y\in c(x)\setminus[u,u']$ are independent of eachother, allowing us to simply accumulate them while iterating through $u'\in\tau(u)\setminus\{u\}$.

\section{Anti Tai mapping problem}\label{s3}

We first present a quadratic algorithm for computing optimal anti Tai mapping in the case when one of the trees consists of a single root-to-leaf path. It generalizes an algorithm from \cite{trajan} which computes optimal anti Tai mapping in the case when both trees consist of a single root-to-leaf path.

\begin{tm}\label{thmopt}
  Optimal anti Tai mapping $M$ of path $[u,v]$ and tree rooted in $x$ satisfies
  $$w(M)=\gamma (x,u,v):=w(v, x) +
  \begin{cases}
    \displaystyle\sum_{y\in c(x)}\gamma(y,u,v) & v=u \\
    \max
    \begin{cases}
      \displaystyle\sum_{y\in c(x)}\gamma(y,u,v) \\
      \gamma(x, u, p(v))
    \end{cases} & v\ne u
  \end{cases}
  $$
  \proof
  Let $(k,l)$ be edge selected at any point of the dynamic program. Then any element of $[u, k]\times \tau(l)$ forms an anti Tai mapping with $(k,l)$. Inductively, any set found by the dynamic program is anti Tai mapping. Conversely, let $M$ be an anti Tai mapping. If $(k,l)\in M$ then $(x,y)\notin M$ for all $(x,y)\in ([k,v]\setminus\{k\})\times (\tau(l)\setminus \{l\})$ and $ (x,y)\notin M$ for all $(x,y)\in ([u,k]\setminus \{k\})\times ([x,l]\setminus\tau(l))$ so inductively there is a sequence of steps corresponding to $M$.
  \qed
\end{tm}

Running time of the algorithm presented by theorem \ref{thmopt} is $O(|V_1||V_2|)$. We see this by noting that the dynamic program has $O(|V_1||V_2|)$ states and that for a fixed $y\in [u,v]$, computing all $\gamma(x,u,y)$ amounts to a single tree traversal.

Theorem \ref{thmdp} therefore suggests a natural generalization of theorem \ref{thmopt} to the case of two arbitrary trees by simply replacing $\alpha$ with $\gamma$.

\begin{cor}\label{mainc}
  $\displaystyle\max_{M \text{anti Tai mapping}}w(M)\ge\max \begin{cases}f(r(G_1), r(G_2)) \\ f(r(G_2), r(G_1))\end{cases}$, where
  $$f(u,v)=\max\begin{cases}\displaystyle\max_{\substack{u'\in \tau (u)\setminus \{u\} \\ v'\in \tau (v)\setminus \{v\}}}\left[ f(v', u') + \sum_{x\in [u,p(u')]} \sum_{y\in c(x) \setminus [u,u']} \gamma(y, v, v')\right] \\ \displaystyle\max_{v'\in \tau(v)} \gamma (u, v, v')\end{cases}$$

\end{cor}

\subsection{Implementation}

The algorithm of corollary \ref{mainc} can be implemented with asymptotic running time of $O(|V_1|^2|V_2|^2)$, analogously to the optimal si-antimatching algorithm of theorem \ref{thmdp}. We provide a sample implementation at \cite{github}. It is also worth noting, for practical purposes, that those algorithms admit a straightforward parallelization. Specifically, all $v'\in\tau(v)\setminus\{v\}$ can be checked in parallel.

\subsection{Anti Tai mapping on DAGs}
In the following section, we will make two straightforward observations regarding generalizations of anti Tai mapping to more general directed acyclic graphs (DAGs). First note that theorem \ref{thmopt} cannot be extended to DAGs since it depends on sets of descendants of distinct children being disjoint. Observation~\ref{tm:DAG_lowerbound} notes that we can, however, do somewhat better than falling back to computing anti Tai mapping on pairs of paths (as is done in \cite{trajan}) by extending a path to an arbitrary topological ordering. Observation~\ref{tm:DAG_antichain} relates maximum-weight anti Tai mapping to maximum weight antichain in the graph constructed by the product of 
two DAGs solving the problem when one of the DAGs is a path in polynomial time.

\begin{obs}\label{tm:DAG_lowerbound}
  Let $G_1=(V_1,E_1)$ and $G_2=(V_2, E_2)$ be directed acyclic graphs, $V_1=\{u_1,...,u_n\}$, $u_1<...<u_n$, $(v_1, v_2,...,v_m)$ arbitrary topological ordering of vertices of $G_2$. Then optimal anti Tai mapping $M$ of $G_1$ and $G_2$ satisfies $w(M)\ge\delta(u_n,v_1)$ where $$\delta (u_i, v_j):=w(u_i,v_j)+\max\begin{cases}\delta(u_{i-1},v_{j}) & i\ne 1 \\ \delta(u_{i},v_{j+1}) & j\ne m\end{cases}$$
\end{obs}
The observation follows from the fact that by definition of topological ordering we have $i<j\implies v_i<v_j\text{ or }v_i$ incomparable with $v_j$.

\begin{obs}\label{tm:DAG_antichain}
  Let $G_1=(V_1,E_1)$ and $G_2=(V_2, E_2)$ be directed acyclic graphs, $V_1=\{u_1,...,u_n\}$, $u_1<...<u_n$. Let $G=(V_1\times V_2, E)$ where $E=\{((u,v),(u',v')) : u<u' \text{ and } v<v'\}$. Then the weight of optimal anti Tai mapping of $G_1$ and $G_2$ is equal to the weight of a maximum weight antichain in $G$.
\end{obs}
The observation follows from the fact that $G$ is a transitively closed directed acyclic graph and $\{x,y\}$  is an anti Tai mapping if and only if  $(x,y)\notin E\text{ and }(y,x)\notin E$.

\section{Conclusion}\label{s5}
In this paper we consider the problem of generating cutting planes for the natural integer programming formulation of the Tai mapping problem. Our cutting planes are based on finding a maximum-weight clique in a subgraph defined on the set of pairs of nodes of the two given trees, which we call an anti Tai mapping. For the special class of si-antimatching, we give a decomposition theorem that describes its precise structure and hence allows us to use dynamic programming to find the maximum-weight si-antimatching in  $O(|V_1|^2|V_2|^2)$ time. Inspired by this result, we also obtain a dynamic program that provides a polynomially computable lower bound  on the maximum-weight anti Tai mapping. Whether the latter problem is NP-hard remains an interesting open question.




\bibliography{mybibfile}

\end{document}